\documentclass[11pt]{article}

\usepackage{amsmath,amsfonts,amsthm,amssymb,bm} 

\usepackage{fullpage}
\usepackage{ifthen}
\usepackage{xspace}

\newtheorem{theorem}{Theorem}[section]
\newtheorem{corollary}[theorem]{Corollary}
\newtheorem{lemma}[theorem]{Lemma}
\newtheorem{proposition}[theorem]{Proposition}
\newtheorem{definition}{Definition}[section]
\newtheorem{example}{Example}[section]

\usepackage{pstricks,pst-node}

\DeclareMathOperator{\argmin}{argmin}
\DeclareMathOperator{\argmax}{argmax}

\providecommand{\SET}[1]{\ensuremath{\{ #1 \}}\xspace}
\providecommand{\Abs}[1]{\ensuremath{| #1 |}\xspace}
\providecommand{\Set}[2]{\ensuremath{\SET{#1 \mid #2}}\xspace}

\providecommand{\half}{\ensuremath{\frac{1}{2}}\xspace}

\providecommand{\Kth}[1]{\ensuremath{{#1}^{\rm th}}}
\providecommand{\Ceiling}[1]{\ensuremath{\lceil {#1} \rceil}\xspace}
\providecommand{\Floor}[1]{\ensuremath{\lfloor {#1} \rfloor}\xspace}

\providecommand{\PROB}{\ensuremath{{\rm Pr}}\xspace}
\providecommand{\Prob}[2][]{\ensuremath{%
\ifthenelse{\equal{#1}{}}{\PROB[#2]}{\PROB_{#1}[#2]}}\xspace}

\providecommand{\Expect}[2][]{\ensuremath{%
\ifthenelse{\equal{#1}{}}{\mathbb{E}}{\mathbb{E}_{#1}}%
\left[#2\right]}\xspace}

\newcommand{\Event}[1][]{\ensuremath{%
\ifthenelse{\equal{#1}{}}{\mathcal{E}}{\mathcal{E}_{#1}}}\xspace}

\def\QED{{\phantom{x}} \hfill \ensuremath{\rule{1.3ex}{1.3ex}}}

\newcommand{\e}{\mathrm{e}}
\newcommand{\dd}{\mathrm{d}}

\newcommand{\RR}{\mathbb{R}}

\newcommand{\NN}{\mathbb{N}}

\newcommand{\QQ}{\mathbb{Q}}

\newcommand{\set}[1]{\left\{ #1 \right\}}

\newcommand{\intersect}{\cap}

\newcommand{\floor}[1]{\lfloor {#1} \rfloor}
\newcommand{\ceil}[1]{\lceil {#1} \rceil}
\renewcommand{\hat}{\widehat}

\newcommand{\eqdef}{:=}

\newcommand{\NUMCAND}{\ensuremath{n}\xspace}
\newcommand{\NUMCANDH}{\ensuremath{\hat{n}}\xspace}

\newcommand{\NUMCANDZ}{\ensuremath{n_0}\xspace}

\newcommand{\VPROB}{\ensuremath{\bm{q}}\xspace}
\newcommand{\VProb}[1]{\ensuremath{q_{#1}}\xspace}

\newcommand{\DIST}{\ensuremath{d}\xspace}
\newcommand{\Dist}[2]{\ensuremath{d_{#1,#2}}\xspace}
\newcommand{\Cost}[1]{\ensuremath{c_{#1}}\xspace}
\newcommand{\Dstr}[1]{\ensuremath{D(#1)}\xspace}

\newcommand{\Winner}[1]{\ensuremath{w(#1)}\xspace}
\newcommand{\Opt}[1]{\ensuremath{o(#1)}\xspace}
\newcommand{\OPT}{\ensuremath{\hat{o}}\xspace}
\newcommand{\Score}[1]{\ensuremath{{\sigma}(#1)}\xspace}
\newcommand{\Pos}[2][]{\ensuremath{%
    \ifthenelse{\equal{#1}{}}{\pi(#2)}{\pi_{#1}(#2)}}\xspace}

\newcommand{\SCF}[1][]{\ensuremath{
    \ifthenelse{\equal{#1}{}}{g}{g_{#1}}}\xspace}
\newcommand{\ScF}[2][]{\ensuremath{\SCF[#1](#2)}\xspace}
\newcommand{\SCFP}[1][]{\ensuremath{%
    \ifthenelse{\equal{#1}{}}{g'}{g'_{#1}}}\xspace}

\newcommand{\SCC}{\ensuremath{c}\xspace} 

\newcommand{\Ball}[2]{\ensuremath{B(#1,#2)}\xspace}

\newcommand{\VS}{\ensuremath{\mathcal{V}}\xspace} 
\newcommand{\AV}{\ensuremath{\Omega}\xspace}
\newcommand{\VP}{\ensuremath{\omega}\xspace} 
\newcommand{\VPP}{\ensuremath{\omega'}\xspace}

\newcommand{\Vol}[1]{\ensuremath{H(#1)}\xspace} 
\newcommand{\tr}{\ensuremath{\hat{r}}\xspace} 
\newcommand{\SBM}{\ensuremath{{\mu}}\xspace} 
\newcommand{\SB}{\ensuremath{{S}}\xspace} 
\newcommand{\LB}{\ensuremath{{T}}\xspace} 
\newcommand{\SBC}{\ensuremath{{s}}\xspace} 

\newcommand{\OC}{\ensuremath{{A}}\xspace} 
\newcommand{\FC}{\ensuremath{{B}}\xspace} 
\newcommand{\OCS}{\ensuremath{M}\xspace} 

\newcommand{\OCM}{\ensuremath{{\alpha}}\xspace} 
\newcommand{\FCM}{\ensuremath{{\beta}}\xspace} 
\newcommand{\OCMU}{\ensuremath{{\hat{\alpha}}}\xspace} 

\newcommand{\OCMC}{\ensuremath{a}\xspace} 
\newcommand{\FCMC}{\ensuremath{b}\xspace} 

\newcommand{\FCP}{\ensuremath{{x}}\xspace} 
\newcommand{\FCPP}{\ensuremath{{x'}}\xspace}
\newcommand{\FCPH}{\ensuremath{{\hat{x}}}\xspace}
                                
\newcommand{\CS}{\ensuremath{{C}}\xspace} 
\newcommand{\OCC}{\ensuremath{\CS_{\OC}}\xspace} 
\newcommand{\SCALE}{\ensuremath{s}\xspace} 
\newcommand{\APX}{\ensuremath{c}\xspace} 

\newcommand{\ScFLB}{\ensuremath{{\ell}}\xspace} 
\newcommand{\ScFUB}{\ensuremath{{u}}\xspace}
\newcommand{\ScFSm}{\ensuremath{{\gamma}}\xspace} 
\newcommand{\ScFSmC}{\ensuremath{{\Gamma}}\xspace} 
\newcommand{\ScFSmCp}{\ensuremath{{\Gamma'}}\xspace}

\begin{document}
\def\ourtitle{On the Distortion of Voting \\ with Multiple Representative Candidates}
\title{\ourtitle}
\author{Yu Cheng \\ Duke University \and
Shaddin Dughmi \\ University of Southern California \and
David Kempe \\ University of Southern California}
\date{}

\maketitle

\begin{abstract}
We study positional voting rules when candidates and voters
are embedded in a common metric space,
and cardinal preferences are naturally given by distances in the metric space. 
In a positional voting rule, each candidate receives a score from each
ballot based on the ballot's rank order;
the candidate with the highest total score wins the election. 
The \emph{cost} of a candidate is his sum of distances to all voters, and
the distortion of an election is the ratio between the cost of the
elected candidate and the cost of the optimum candidate.
We consider the case when candidates are \emph{representative} of the population,
in the sense that they are drawn i.i.d.~from the population of the voters,
and analyze the expected distortion of positional voting rules.
  
Our main result is a clean and tight characterization of
positional voting rules that have constant expected distortion
(independent of the number of candidates and the metric space).
Our characterization result immediately implies constant expected distortion for 
Borda Count and elections in which each voter approves a constant \emph{fraction} of all candidates.
On the other hand, we obtain super-constant expected distortion for
Plurality, Veto, and approving a constant \emph{number} of candidates.
These results contrast with previous results on voting with metric
preferences: 
When the candidates are chosen adversarially,
all of the preceding voting rules have distortion linear in
the number of candidates or voters.
Thus, the model of representative candidates allows us to
distinguish voting rules which seem equally bad in the worst case.

\end{abstract}

\section{Introduction} \label{sec:introduction}

In light of the classic impossibility results for axiomatic approaches
to social choice \cite{arrow:social-choice} and voting
\cite{gibbard:manipulation,satterthwaite:voting},
a fruitful approach has been to treat voting as an implicit
optimization problem of finding the ``best'' candidate for the
population in aggregate \cite{BCHLPS:utilitarian:distortion,caragiannis:procaccia:voting,procaccia:approximation:gibbard,procaccia:rosenschein:distortion}.
Using this approach, voting systems can be compared based on how much
they \emph{distort} the outcome,
in the sense of leading to the election of suboptimal candidates.
A particularly natural optimization objective is the \emph{sum of distances}
between voters and the chosen candidate in a suitable metric space
\cite{anshelevich:ordinal,anshelevich:bhardwaj:postl,anshelevich:postl:randomized,goel:krishnaswamy:munagala}. 
The underlying assumption is that the closer a candidate is to a voter,
the more similar their positions on key questions are.
Because proximity implies that the voter would benefit from the
candidate's election,
voters will rank candidates by increasing distance, 
a model known as \emph{single-peaked preferences}
\cite{black:rationale,downs:democracy,black:committees-elections,moulin:single-peak,merrill:grofman,barbera:gul:stacchetti,richards:richards:mckay,barbera:social-choice}.

Even in the absence of strategic voting,
voting systems can lead to high distortion in this setting,
because they typically allow only for communication of \emph{ordinal} preferences\footnote{Of course, it is also highly
  questionable that voters would be able to quantify distances in a
  metric space sufficiently accurately, in particular given that the
  metric space is primarily a modeling tool rather than an actual
  concrete object.}, 
i.e., rankings of candidates \cite{boutilier:rosenschein:incomplete}.
In a beautiful piece of recent work,
Anshelevich et al.~\cite{anshelevich:bhardwaj:postl} showed that this approach
can draw very clear distinctions between voting systems:
some voting systems (in particular, Copeland and related systems)
have distortion bounded by a small constant,
while most others (including Plurality, Veto, $k$-approval, and Borda Count)
have unbounded distortion,
growing linearly in the number of voters or candidates.

The examples giving bad distortion typically have the property
that the candidates are not ``representative'' of the voters.
Anshelevich et al.~\cite{anshelevich:bhardwaj:postl} show more positive
results when there are no near-ties for first place in any voter's
ranking.
Cheng et al.~\cite{CDK17:of-the-people} propose instead a
model of representativeness in which the candidates are drawn randomly
from the population of voters;
under this model, 
they show smaller constant
distortion bounds than 
the worst-case bounds for majority voting with $\NUMCAND = 2$
candidates.
Cheng et al.~\cite{CDK17:of-the-people} left as an open question the analysis of
the distortion of voting systems for $\NUMCAND \geq 3$ representative
candidates.

In the present work, we study the distortion of \emph{positional voting systems} with $\NUMCAND \geq 3$ representative candidates.
Informally (formal definitions of all concepts are given in
Section~\ref{sec:preliminaries}),
a positional voting system is one in which each voter writes down an
ordering of candidates, and the system assigns a score to each candidate based solely on his\footnote{For
  consistency, we always use male pronouns for candidates and female
  pronouns for voters. } position in the voter's ordering. The map from positions to scores is known as the \emph{scoring rule} of the voting system, and for $\NUMCAND$ candidates is a function $\SCF[\NUMCAND]:\set{0,\ldots,\NUMCAND-1} \to \RR_{\ge 0}$. 
The total score of a candidate is the sum of scores he
obtains from all voters,
and the winner is the candidate with maximum total score.
The most well-known 
explicitly positional voting
system is Borda Count~\cite{borda:elections},
in which $\ScF[\NUMCAND]{i} = n-i$ for all $i$.
Many other systems
are naturally cast in this framework, including
Plurality (in which voters give 1 point to their first choice only)
and Veto (in which voters give 1 point to all but their last choice).

In analyzing positional voting systems, we assume that voters are not strategic,
i.e., they report their true ranking of candidates based on proximity in the metric space.
This is in keeping with the line of work on analyzing the distortion
of social choice functions,
and avoids issues of game-theoretic modeling and equilibrium
existence or selection (see, e.g., \cite{feldman:fiat:golomb})
which are not our focus.


As our main contribution, we characterize when a positional voting system is guaranteed to have constant
distortion, regardless of the underlying metric space of voters and
candidates, and regardless of the number \NUMCAND of candidates that
are drawn from the voter distribution.
The characterization relies almost entirely on the
``limit voting system.''
By normalizing both the scores and the candidate index to lie in
$[0,1]$ (we associate the \Kth{i} out of \NUMCAND
candidates with his quantile $\frac{i}{\NUMCAND-1} \in [0,1]$),
we can take a suitable limit \SCF of the scoring functions
\SCF[\NUMCAND] as $\NUMCAND \to \infty$.

Our main result (Corollary~\ref{cor:generalized-borda-constant} in
Section~\ref{sec:generalized-borda}) states the following:
(1) If \SCF is not constant on the open interval $(0,1)$,
then the voting system has constant distortion.
(2) If \SCF is a constant other than 1 on the open interval $(0,1)$,
then the voting system does not have constant distortion.
The only remaining case is when $\SCF \equiv 1$ on $(0,1)$.
In that case, the rate of convergence of \SCF[\NUMCAND] to \SCF
matters, and a precise characterization is given by
Theorem~\ref{thm:main-generalized-borda}.

As direct applications of our main result, we obtain that
Borda Count and $k$-approval for $k=\Theta(\NUMCAND)$ representative
candidates have constant distortion;
on the other hand, Plurality, Veto, the Nauru Dowdall method
(see Section~\ref{sec:preliminaries}),
and $k$-approval for $k=O(1)$ have super-constant distortion.
In fact, it is easy to adapt the proof of
Theorem~\ref{thm:main-generalized-borda}
to show that the distortion of Plurality, Veto, and $O(1)$-approval,
even with representative candidates, is $\Omega(\NUMCAND)$.

Our results provide interesting contrasts to the results of 
Anshelevich et al.~\cite{anshelevich:bhardwaj:postl}.
Under adversarial candidates, all of the above-mentioned voting rules
have distortion $\Omega(\NUMCAND)$;
the focus on representative candidates allowed us to distinguish the
performance of Borda Count and $\Theta(\NUMCAND)$-approval from that
of the other voting systems.
Thus, an analysis in terms of representative candidates allows us to
draw distinctions between voting systems which in a worst-case setting
seem to be equally bad.

As a by-product of the proof of our main theorem,
in Lemma~\ref{lem:linear-distortion},
we show that \emph{every} voting system (positional or otherwise) has distortion $O(\NUMCAND)$
with representative candidates.
Combined with the lower bound alluded to above,
this exactly pins down the distortion of Plurality, Veto, and
$O(1)$-approval with representative candidates to $\Theta(\NUMCAND)$.
For Veto, this result also contrasts with the worst-case bound of
Anshelevich et al.~\cite{anshelevich:bhardwaj:postl}, which showed that the distortion
can grow unboundedly even for $\NUMCAND=3$ candidates.


 \section{Preliminaries}
\label{sec:preliminaries}

\subsection{Voters, Metric Space, and Preferences}
The voters/candidates are embedded in a closed metric space $(\AV, \DIST)$,
where \Dist{\VP}{\VPP} is the distance between points $\VP, \VPP \in \AV$.
The distance captures the dissimilarity in opinions between voters
(and candidates) --- the closer two voters or candidates are,
the more similar they are.
The distribution of voters in \AV is denoted by 
the (measurable) density function \VProb{\VP}.
We allow for \VPROB to have point masses.\footnote{Since the continuum model allows for point masses,
it subsumes finite sets of voters.
Changing all our results to finite or countable voter
sets is merely cosmetic.}
Unless there is no risk of confusion, we will be careful to
distinguish between a location $\VP \in \AV$ and a specific
voter $j$ or candidate $i$ at that location.
We apply \DIST equally to locations/voters/candidates.

We frequently use the standard notion of a \emph{ball} 
$\Ball{\VP}{r} \eqdef \Set{\VPP}{\Dist{\VP}{\VPP} \leq r}$
in a metric space.
For balls (and other sets) $B$, we write
$\VProb{B} \eqdef \int_{\VP \in B} \VProb{\VP} \dd\VP$.

An election is run between $\NUMCAND \geq 2$ candidates
according to rules defined in Section~\ref{sec:prelim:scoring-rules}.
The \NUMCAND candidates are assumed to be representative of the
population, in the sense that their locations are drawn i.i.d.~from
the distribution \VPROB of voters.
 
Each voter ranks the \NUMCAND candidates $i$ by non-decreasing
distance from herself in $(\AV, \DIST)$.
Ties are broken arbitrarily, but consistently%
\footnote{Our results do not depend on specific tie breaking rules.},
meaning that all voters at the same location have the same ranking.
We denote the ranking of a voter $j$ or a location \VP over candidates
$i$ by \Pos[j]{i} or \Pos[\VP]{i}.
The distance-based ranking assumption means that
$\Pos[\VP]{i} < \Pos[\VP]{i'}$ implies that
$\Dist{\VP}{i} \leq \Dist{\VP}{i'}$ and
$\Dist{\VP}{i} < \Dist{\VP}{i'}$ implies that
$\Pos[\VP]{i} < \Pos[\VP]{i'}$.
As mentioned in the introduction, we assume that voters are not strategic;
i.e., they express their true ranking of candidates based on proximity
in the metric space.


\subsection{Social Cost and Distortion}
Candidates are ``better'' if they are closer to voters on average.
The \emph{social cost} of a candidate (or location) $i$ is
\begin{align*}
\Cost{i} & = \int_{\VP} \Dist{\VP}{i} \VProb{\VP} \dd\VP.
\end{align*}

The socially optimal candidate among the set \CS of candidates running
is denoted by $\Opt{\CS} \eqdef \argmin_{i \in \CS} \Cost{i}$.
The overall optimal location is denoted by
$\OPT \in \argmin_{\VP \in \AV} \Cost{\VP}$,
which is any 1-median of the metric space.
(If there are multiple optimal locations, consider one of them fixed
arbitrarily.)
The $\argmin$ always exists,
because the metric space is assumed to be closed,
and the cost function is continuous and bounded below by 0.
Note that it is not necessary that there be any voters located at
\OPT.

Based on the votes, a voting system will determine a \emph{winner}
\Winner{\CS} for the set \CS of candidates,
who will often be different from \Opt{\CS}.
The \emph{distortion} measures how much worse the winner is
than the optimum
\begin{align*}
\Dstr{\CS} & = \frac{\Cost{\Winner{\CS}}}{\Cost{\Opt{\CS}}}.
\end{align*}  
We are interested in the \emph{expected distortion} of positional
voting systems under i.i.d.~random candidates, i.e.,
\begin{align*}
\Expect[\CS \stackrel{\text{ i.i.d. }}{\sim} \VPROB]{\Dstr{\CS}}.
\end{align*}

Our distortion bounds are achieved by lower-bounding
$\Cost{\Opt{\CS}} \geq \Cost{\OPT}$.
A particularly useful quantity in this context is the fraction of
voters outside a ball of radius $r$ around \OPT, which we denote by
$\Vol{r} \eqdef 1 - \VProb{\Ball{\OPT}{r}}$.
The following lemma captures some useful simple facts that we use:

\begin{lemma} \label{lem:basic-facts}
\begin{enumerate}
  \item For any candidate or location $i$,
    \begin{align}
      \Cost{i} & \leq \Cost{\OPT} + \Dist{i}{\OPT}.
                   \label{eqn:triangle-bound}
    \end{align}
  \item The cost of any candidate or location $i$ can be written as
    \begin{align}
      \Cost{i} & = \int_0^{\infty} (1 - \VProb{\Ball{i}{r}}) \dd r.
                 \label{eqn:expect-rewrite}
    \end{align}
    
  \item For all $r \geq 0$, the cost of the optimum location \OPT is
                   lower-bounded by
    \begin{align}
      \Cost{\OPT} & \geq r \Vol{r}.
                 \label{eqn:opt-lower-bound}
    \end{align}
\end{enumerate}
\end{lemma}

\begin{emptyproof}
\begin{enumerate}
\item The proof of the first inequality simply applies the triangle
  inequality under the integral:
  \[
    \Cost{i} \; = \; \int_{\VP} \Dist{\VP}{i} \VProb{\VP} \dd\VP
    \; \leq \; \int_{\VP} (\Dist{\VP}{\OPT} + \Dist{i}{\OPT}) \VProb{\VP} \dd\VP
    \; = \; \Cost{\OPT} + \Dist{i}{\OPT}.
  \]
\item For the second equation, observe that
  $\Cost{i} = \Expect[\VP \sim \VPROB]{\Dist{i}{\VP}}$,
  and the expectation of any non-negative random variable $X$ can be
  rewritten as 
  $\Expect{X} = \int_0^{\infty} \Prob{X \geq x} \dd x$.
\item For the third inequality, we apply the previous part with $i=\OPT$,
  and lower bound
  \begin{align*}
    \int_0^{\infty} \Prob[\VP \sim \VPROB]{\Dist{\OPT}{\VP} \geq x} \dd x
    & \; = \; \int_0^r \Prob[\VP \sim \VPROB]{\Dist{\OPT}{\VP} \geq x} \dd x
          + \int_r^{\infty} \Prob[\VP \sim \VPROB]{\Dist{\OPT}{\VP} \geq x} \dd x
 \\ & \; \geq \; \int_0^r \Prob[\VP \sim \VPROB]{\Dist{\OPT}{\VP} \geq r} \dd x
         + \int_r^{\infty} 0 \; \dd x
 \\ & \; = \; r \cdot \Vol{r}. \QED
  \end{align*}
\end{enumerate}
\end{emptyproof}

\subsection{Positional Voting Systems and Scoring Rules}
\label{sec:prelim:scoring-rules}

We are interested in \emph{positional voting systems}.
Such systems are based on \emph{scoring rules}:
voters give a ranking of candidates,
and with each position is associated a score.

\begin{definition}[Scoring Rule] \label{def:scoring-rule}
A scoring rule for \NUMCAND candidates is a non-increasing function
$\SCF[\NUMCAND] : \SET{0, \ldots, \NUMCAND-1} \to [0,1]$ with
$\ScF[\NUMCAND]{0} = 1$ and $\ScF[\NUMCAND]{\NUMCAND-1} = 0$.
\end{definition}

\begin{definition}[Positional Voting System] \label{def:voting-system}
A positional voting system is a sequence of scoring rules \SCF[\NUMCAND], one for each number of candidates $\NUMCAND = 1, 2, \ldots$.
\end{definition}

The interpretation of \SCF[\NUMCAND] is that if voter $j$ puts a candidate $i$ in position \Pos[j]{i} on her ballot, then $i$ obtains \ScF[\NUMCAND]{\Pos[j]{i}} points from $j$. The total score of candidate $i$ is then
\begin{align*}
\Score{i} & = \int_{\VP} \ScF[\NUMCAND]{\Pos[\VP]{i}} \VProb{\VP} \dd\VP.
\end{align*}
The winning candidate is one with highest total score, i.e.,
for a set \CS of \NUMCAND candidates,
$\Winner{\CS} \in \argmax_{i \in \CS} \Score{i}$;
again, ties are broken arbitrarily, and our results do not depend on
tie breaking.

The restriction to monotone non-increasing scoring rules is standard
when studying positional voting systems.
One justification is that in any positional voting system violating
this restriction, truth-telling is a dominated strategy,
rendering such a system uninteresting for most practical purposes.
Given this restriction, the assumption that $\ScF[\NUMCAND]{0} = 1$
and $\ScF[\NUMCAND]{\NUMCAND-1} = 0$ is without loss of generality,
because a score-based rule is invariant under affine transformations.

Next, we want to capture the notion that a positional voting system
is ``consistent'' as we vary the number of candidates \NUMCAND.
Intuitively, we want to exclude contrived voting systems such as
``If the number of candidates is even, then use Borda Count;
otherwise use Plurality voting.''
This is captured by the following definition.

\begin{definition}[Consistency] \label{def:consistency}
Let \VS be a positional voting system with scoring rules
$\Set{\SCF[\NUMCAND]}{\NUMCAND \in \NN}$.
We say that \VS is \emph{consistent} if there exists a function 
$\SCF:\QQ \intersect [0,1] \to [0,1]$ such that for each rational quantile $x \in [0,1]$ and accuracy
parameter $\epsilon >0$,
there exists a threshold \NUMCANDZ such that
$\ScF[\NUMCAND]{\floor{x(\NUMCAND-1)}} \geq \ScF{x} - \epsilon$ and $\ScF[\NUMCAND]{\ceil{x(\NUMCAND-1)}} \leq \ScF{x} +\epsilon$
for all 
$\NUMCAND \geq \NUMCANDZ$.  We call \SCF the \emph{limit scoring rule of \VS}.
\end{definition}

Intuitively, this definition says that the sequence of scoring rules $\SCF[\NUMCAND]$ is consistent with a single scoring rule $\SCF$ in the limit. Using the fact that $\SCF[\NUMCAND]$ is monotone non-increasing for each $\NUMCAND$, it can be shown that $\SCF$ is also monotone non-increasing.  We note that $\SCF[\NUMCAND]$ converges pointwise to $\SCF$ in a precise and natural sense. Formally, when $x \in [0,1]$ is rational, there exists an infinite sequence of integers $\NUMCAND$ with $\floor{x(\NUMCAND-1)}=\ceil{x(\NUMCAND-1)}=x(\NUMCAND-1)$, and consistency implies that $\ScF{x}$ must equal the limit of $\ScF[\NUMCAND]{x(\NUMCAND-1)}$ for that sequence of values of $\NUMCAND$. Therefore the limit scoring rule $\SCF$ is uniquely defined if it exists.

All positional voting systems we are aware of are consistent according to Definition~\ref{def:consistency}.

\begin{example} \label{exm:common-positional-systems}
To illustrate the notion of a consistent positional voting system,
consider the following examples, encompassing most well-known scoring rules.

\begin{itemize}
\item In Plurality voting with \NUMCAND candidates,
  $\ScF[\NUMCAND]{0} = 1$ and $\ScF[\NUMCAND]{k} = 0$ for all $k > 0$.
  The limit scoring rule is $\ScF{0} = 1$ and $\ScF{x} = 0$ for all
  $x > 0$.

\item In Veto voting with \NUMCAND candidates,
  $\ScF[\NUMCAND]{k} = 1$ for all $k < \NUMCAND-1$
  and $\ScF[\NUMCAND]{\NUMCAND-1} = 0$. 
  The limit scoring rule is $\ScF{x} = 1$ for all $x < 1$
  and $\ScF{1} = 0$.

\item In $k$-approval voting with constant $k$,
  we have $\ScF[\NUMCAND]{k'} = 1$ for $k' \leq \min(k-1, \NUMCAND-1)$,
  and $\ScF[\NUMCAND]{k'} = 0$ for all other $k'$.
  The limit scoring rule is $\ScF{0} = 1$ and $\ScF{x} = 0$ for all
  $x > 0$, i.e., the same as for Plurality voting.
  (This relies on $k$ being constant, or more generally, $k = o(\NUMCAND)$.)

\item In $k$-approval voting with linear $k$,
  there exists a constant $\gamma \in (0,1)$ with
  $\ScF[\NUMCAND]{k} = 1$ for all $k \leq \gamma \NUMCAND$,
  and $\ScF[\NUMCAND]{k} = 0$ for all larger $k$.
  The limit scoring rule is $\ScF{x} = 1$ for $x \leq \gamma$ and
  $\ScF{x} = 0$ for $x > \gamma$.

\item The Borda voting rule has
  $\ScF[\NUMCAND]{k} = 1 - \frac{k}{\NUMCAND-1}$ (after normalization).
  The limit scoring rule is $\ScF{x} = 1-x$.

\item The Dowdall method used in Nauru
  \cite{fraenkel:grofman:borda-count,reilly:south-seas}
  has $\ScF[\NUMCAND]{k} = 1/(k+1)$.
  After normalization, the rule becomes 
  $\ScF[\NUMCAND]{k} = \frac{1}{\NUMCAND-1} \cdot
  (\frac{\NUMCAND}{k+1} - 1)$.
  The limit scoring rule is $\ScF{0} = 1$ and $\ScF{x} = 0$
  for all $x > 0$, i.e., the same as for Plurality voting.
  This is because for every constant quantile $x$,
  the score of the candidate at $x$ is
  $\frac{1}{\NUMCAND - 1} \left(\frac{1}{x} - 1\right)
  \stackrel{\NUMCAND \to \infty}{\to} 0$.
\end{itemize}
\end{example}

\section{The Main Characterization Result}
\label{sec:generalized-borda}

In this section, we state and prove our main theorem,
characterizing positional voting systems with constant distortion.

\begin{theorem} \label{thm:main-generalized-borda}
Let \VS be a positional voting system with a sequence \SCF[\NUMCAND]
of scoring rules for $\NUMCAND = 1, 2, \ldots$.
Then, \VS has constant expected distortion if and only if
there exist constants \NUMCANDZ and $y \in (0, 1)$ such that
for all $\NUMCAND \geq \NUMCANDZ$,
 \begin{align}
 y \cdot \sum_{k=0}^{\Ceiling{y (\NUMCAND - 1)}-1} \left( \ScF[\NUMCAND]{k} - \ScF[\NUMCAND]{\Ceiling{y (\NUMCAND - 1)}} \right)
 & \; > \; (1-y) \cdot \sum_{k=\NUMCAND - \Ceiling{y (\NUMCAND - 1)}}^{\NUMCAND-1} \left( 1 - \ScF[\NUMCAND]{k} \right).
 \label{eqn:limit-condition}
 \end{align}
\end{theorem}

We prove Theorem~\ref{thm:main-generalized-borda} in Sections~\ref{sec:sufficiency} (sufficiency) and \ref{sec:necessity} (necessity).
Condition \eqref{eqn:limit-condition} is quite unwieldy.
In most cases of practical interest, we can use
Corollary~\ref{cor:generalized-borda-constant}. 

\begin{corollary} \label{cor:generalized-borda-constant}
Let \VS be a consistent positional voting system with limit scoring rule \SCF.
\begin{enumerate}
\item If \SCF is not constant on the open interval $(0,1)$,
then \VS has constant expected distortion.
\item If \SCF is equal to a constant other than 1 on the open interval
  $(0,1)$, then \VS does not have constant expected distortion.
\end{enumerate}
\end{corollary}

Corollary~\ref{cor:generalized-borda-constant} is proved in Section~\ref{sec:generalized-borda-constant-proof}. 

The constant in Theorem~\ref{thm:main-generalized-borda} and
Corollary~\ref{cor:generalized-borda-constant} depends on \VS,
but not on the metric space or the number of candidates.
Corollary~\ref{cor:generalized-borda-constant} has the advantage of
determining constant expected distortion only based on the limit
scoring rule \SCF.
The only case when it does not apply is when $\ScF{x} = 1$
for all $x \in [0,1)$. 
In that case, the higher complexity of
Theorem~\ref{thm:main-generalized-borda} is indeed necessary to
determine whether \VS has constant distortion.
Fortunately, Veto voting is the only rule of practical importance for
which $\ScF{x} \equiv 1$ on $[0,1)$, and it is easily analyzed.


Before presenting the proofs, we apply the characterization to the
positional voting systems from
Example~\ref{exm:common-positional-systems}.
Using the limit scoring rules derived in
Example~\ref{exm:common-positional-systems},
Corollary~\ref{cor:generalized-borda-constant} implies constant
expected distortion for Borda Count and $k$-approval with linear $k =
\Theta(\NUMCAND)$,
and super-constant expected distortion for Plurality,
$k$-approval with $k=o(n)$, and the Dowdall method.
%

This leaves Veto voting, for which it is easy to apply
Theorem~\ref{thm:main-generalized-borda} directly.
Because $\ScF[\NUMCAND]{k} = 1$ for all $k < \NUMCAND-1$,
for any constant $y < 1$ and large enough \NUMCAND,
the left-hand side of \eqref{eqn:limit-condition} is 0,
while the right-hand side is positive.
Hence, \eqref{eqn:limit-condition} can never be satisfied for
sufficiently large \NUMCAND,
implying super-constant expected distortion. 
The proof easily generalizes to show that when voters
can veto $o(\NUMCAND)$ candidates, the distortion is super-constant.

\subsection{Sufficiency} \label{sec:sufficiency}

In this section, we prove that condition~\eqref{eqn:limit-condition}
suffices for constant distortion.
First, because of the monotonicity of \SCF[\NUMCAND],
if \eqref{eqn:limit-condition} holds for $y \in (0,1)$,
then it also holds for all $y' \in [y,1)$.
Now, the high-level idea of the proof is the following:
we define a radius \tr large enough so that the ball \Ball{\OPT}{\tr}
around the socially optimal location \OPT contains a very large
(but still constant) fraction $y$ of all voters,
such that $y$ satisfies \eqref{eqn:limit-condition}.
If the number of candidates \NUMCAND is large enough (a large constant),
standard Chernoff bounds ensure that as $r \geq \tr$ grows large,
most candidates who are running will be from inside \Ball{\OPT}{r}.
In turn, if many candidates inside \Ball{\OPT}{r} are running,
all candidates outside \Ball{\OPT}{3r} are very far down on almost
everyone's ballot, and therefore cannot win.
In particular, Inequality \eqref{eqn:limit-condition} implies that
the total score of an average candidate in \Ball{\OPT}{r}
exceeds the maximum possible total score of a candidate outside \Ball{\OPT}{3r}.
This allows us to bound the expected distortion in terms of the cost
of \OPT.

The case of small \NUMCAND is much easier, since we can treat \NUMCAND
as a constant.
In that case, the following lemma is sufficient.

\begin{lemma} \label{lem:linear-distortion}
  If \NUMCAND candidates are drawn i.i.d.~at random from \VPROB,
  the expected distortion is at most $\NUMCAND+1$.
\end{lemma}
\begin{emptyproof}
The proof illustrates some of the key ideas that will be used later in
the more technical proof for a large number of candidates.
We want to bound
\begin{align*}
\Expect[\CS]{\Cost{\Winner{\CS}}}
& \; \stackrel{\eqref{eqn:triangle-bound}}{\leq} \;
\Cost{\OPT} + \Expect[\CS]{\Dist{\OPT}{\Winner{\CS}}}
\; = \; \Cost{\OPT} + \int_0^{\infty} \Prob[\CS]{\Dist{\OPT}{\Winner{\CS}} \geq r} \dd r.
\end{align*}

In order for a candidate at distance at least $r$ from \OPT to win,
it is necessary that at least one such candidate be running.
By a union bound over the \NUMCAND candidates,
the probability of this event is at most
$\Prob[\CS]{\Dist{\OPT}{\Winner{\CS}} \geq r} \leq \NUMCAND \Vol{r}$, so
\begin{align*}
\Expect[\CS]{\Cost{\Winner{\CS}}}
  & \leq \Cost{\OPT} + \NUMCAND \int_0^{\infty} \Vol{r} \dd r
  \stackrel{\eqref{eqn:expect-rewrite}}{=}
  (\NUMCAND+1) \Cost{\OPT}.
\end{align*}

Lower-bounding the cost of the optimum candidate from $\CS$ in terms of
the overall best location \OPT, the expected distortion is

\begin{align*}
  \Expect[\CS]{\frac{\Cost{\Winner{\CS}}}{\Cost{\Opt{\CS}}}}
 &\; \leq \; \Expect[\CS]{\frac{\Cost{\Winner{\CS}}}{\Cost{\OPT}}}
  \; = \; \frac{1}{\Cost{\OPT}} \Expect[\CS]{\Cost{\Winner{\CS}}}
  \; \leq \; \frac{1}{\Cost{\OPT}} \cdot (\NUMCAND+1) \Cost{\OPT}
  \; = \; \NUMCAND+1. \QED
\end{align*}
\end{emptyproof}

In preparation for the case of large \NUMCAND,
we begin with the following technical lemma,
which shows that whenever \eqref{eqn:limit-condition} holds,
it will also hold when the terms on the left-hand side are ``shifted,''
and the right-hand side can be increased by a factor of 2
(or, for that matter, any constant factor).

\begin{lemma} \label{lem:shifted-comparison}
Assume that there exist $y \in (0,1)$ and $\NUMCAND_0$ such that
\eqref{eqn:limit-condition} holds.
Then, there exists $z_0 \in (\half,1)$ such that
for all $\NUMCAND \geq \NUMCAND_0$,
all $z \geq z_0$,
and all integers $0 \leq m \leq (1-z) \cdot \NUMCAND$,
\begin{align}
& z \cdot \sum_{k=0}^{\Ceiling{z \cdot (\NUMCAND-1)}-1}
    \left( \ScF[\NUMCAND]{m+k} - \ScF[\NUMCAND]{m+\Ceiling{z \cdot (\NUMCAND-1)}} \right)
  \; > \;
2 (1-z) \cdot \sum_{k=\NUMCAND - \Ceiling{z \cdot (\NUMCAND-1)}}^{\NUMCAND-1}
       \left( 1 - \ScF[\NUMCAND]{k} \right).
\label{eqn:shifted-limit-condition}
\end{align}
\end{lemma}

We now flesh out the details of the construction.
By Lemma~\ref{lem:shifted-comparison},
there exists $z_0 \in (\half,1)$ and $\NUMCAND_0$ such that
\eqref{eqn:shifted-limit-condition} holds
for all $z \geq z_0$,
all $\NUMCAND \geq \NUMCAND_0$,
and all integers $0 \leq m \leq (1-z) \cdot \NUMCAND$.
For simplicity of notation, write $z \eqdef z_0$, and
let $\SBM \eqdef (1-\frac{1}{\e}) + \frac{1}{\e} \cdot z \in (z,1)$
and $\NUMCANDH \eqdef \max(\NUMCAND_0, \frac{1}{1-z})$.
Notice that $\SBM, z, \NUMCANDH$ only depend on \VS,
but not on the metric space or number of voters.

Let $\tr \eqdef \inf \Set{r}{\VProb{\Ball{\OPT}{r}} \geq \SBM}$,
so that $\VProb{\Ball{\OPT}{\tr}} \geq \SBM$,
and $\VProb{\Set{\VP}{\Dist{\OPT}{\VP} \geq \tr}} \geq 1-\SBM$.
(Both inequalities hold with equality unless there is a discrete
point mass at distance \tr from \OPT.)


Consider any $r \geq \tr$ and write $\LB \eqdef \Ball{\OPT}{3r}$
and $\SB \eqdef \Ball{\OPT}{r}$, as depicted in Figure~\ref{fig:sufficiency}.
When \NUMCAND candidates are drawn i.i.d.~from \VPROB,
the expected fraction of candidates drawn from outside of \SB is
exactly $\Vol{r} \leq 1-\SBM$.
Let \Event[r] be the event that more than
$(1-z) \NUMCAND$ candidates are from outside \SB.
Lemma~\ref{lem:event-low-probability} uses Chernoff bounds and the
definitions of the parameters to show that \Event[r] happens with
sufficiently small probability;
Lemma~\ref{lem:event-implies-winner} then shows that
unless \Event[r] happens, the distortion is constant.

\begin{figure}
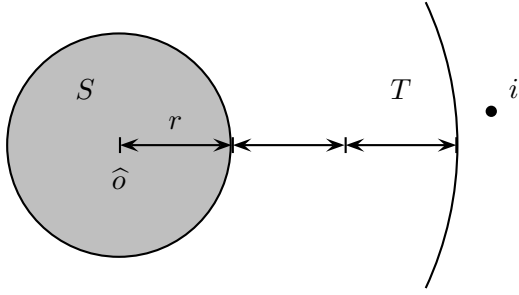

\begin{centering}
\psset{unit=1.5cm,arrowsize=0.1 1}
\pspicture(-1.5,-1.2)(4,1.2)
\pscircle[fillstyle=solid,fillcolor=lightgray](0,0){1}
\psline{|<->}(0,0)(1,0)
\psline{|<->}(1,0)(2,0)
\psline{|<->|}(2,0)(3,0)
\psarc{-}(0,0){3}{-25}{25}

\pscircle[fillstyle=solid,fillcolor=black](3.3,0.3){0.05}
\rput(3.5,0.5){$i$}

\rput(0,-0.3){$\hat o$}
\rput(0.5,0.2){$r$}
\rput(-0.3,0.5){$S$}
\rput(2.5,0.5){$T$}
\endpspicture
\caption{$\LB = \Ball{\OPT}{3r}$ and $\SB = \Ball{\OPT}{r}$. Most of the voters are in \SB. Lemma~\ref{lem:event-implies-winner} states that whenever most of the candidates are from $T$, the winner must come from $T$. The reason is that for any $i \notin \LB$, even an average candidate in $\SB$ beats $i$; in particular, the best candidate from \SB must beat $i$.}
\label{fig:sufficiency}
\end{centering}
\end{figure}

\begin{lemma} \label{lem:event-low-probability}
$\Prob{\Event[r]} \leq \frac{\e}{1-z} \cdot \Vol{r}$.
\end{lemma}

\begin{proof}
By the Chernoff bound $\Prob{Z > (1+\delta) \Expect{Z}}
< \left(\frac{\e^\delta}{(1+\delta)^{1+\delta}} \right)^{\Expect{Z}}$,
applied with $\Expect{Z} = \Vol{r} \cdot \NUMCAND$ and
$\delta = \frac{1-z}{\Vol{r}}-1 > 0$,
the probability of \Event[r] is at most
\begin{align*}
  \Prob{\Event[r]}
  \; & \leq \;
\left(\frac{\e^{\frac{1-z}{\Vol{r}}-1}}{(\frac{1-z}{\Vol{r}})^{\frac{1-z}{\Vol{r}}}} \right)^{\Vol{r} \cdot \NUMCAND}
\; = \;
\left(\frac{\e^{1-z-\Vol{r}}}{(\frac{1-z}{\Vol{r}})^{1-z}} \right)^{\NUMCAND}
  \; \leq \; 
\left(\frac{\e \cdot \Vol{r}}{1-z} \right)^{(1-z) \cdot \NUMCAND}.
\end{align*}

Recall that $\SBM \eqdef (1-\frac{1}{\e}) + \frac{1}{\e} \cdot z \in (z,1)$.
Because $r \geq \tr$, we have that
$\Vol{r} \leq 1-\SBM = \frac{1-z}{\e}$;
in particular, $\frac{\e \cdot \Vol{r}}{1-z} \leq 1$,
so the probability can be upper-bounded by making the exponent
$(1-z) \cdot \NUMCAND$ as small as possible.
Because $\NUMCAND \geq \NUMCANDH \geq \frac{1}{1-z}$,
the exponent is lower-bounded by 1.
Thus, we obtain that the probability of \Event[r] is at most
$\frac{\e}{1-z} \cdot \Vol{r}$.
\end{proof}

\begin{lemma} \label{lem:event-implies-winner}
Whenever \Event[r] does not happen, the winner of the election is from
\Ball{\OPT}{3r}.
\end{lemma}

\begin{emptyproof}
Let $Z \eqdef \Ceiling{z \cdot (\NUMCAND-1)}$.
Assume that exactly $\SBC \geq Z$ out of the \NUMCAND candidates are
drawn from \SB.
Consider a candidate $i \notin \LB$.
We will compare the average number of points of candidates in \SB with
the maximum possible number of points of candidate $i$,
and show that the former exceeds the latter.

\begin{itemize}
\item Each voter $j \notin \SB$ gives at most one point to $i$.
On the other hand, even if $j$ ranks all of \SB in the last \SBC positions,
the total number of points assigned by $j$ to \SB is at least
$\sum_{k=\NUMCAND-\SBC}^{\NUMCAND-1} \ScF[\NUMCAND]{k}$.
The difference between the number of votes to $i$ and the average
number of votes to candidates in \SB is thus at most
\[
  1 - \left( \frac{1}{\SBC} \cdot
    \sum_{k=\NUMCAND-\SBC}^{\NUMCAND-1} \ScF[\NUMCAND]{k}
    \right)
\; = \; \frac{1}{\SBC} \cdot
\sum_{k=\NUMCAND-\SBC}^{\NUMCAND-1} (1-\ScF[\NUMCAND]{k}).
\]

Because no more than a $1-\SBM$ fraction of voters are
strictly outside \SB,
the total advantage of $i$ over an average candidate in \SB resulting
from such voters is at most
\[
\Delta_i \eqdef \frac{1}{\SBC} \cdot (1-\SBM) \cdot 
\sum_{k=\NUMCAND-\SBC}^{\NUMCAND-1} (1-\ScF[\NUMCAND]{k}).
\]

\item Each voter $j \in \SB$ will rank all candidates in \SB
(who are at distance at most $2r$ from her)
ahead of all candidates outside \LB
(who are at distance strictly more than $3r-r = 2r$ from her).

Let $m \geq 0$ be such that $j$ ranks $i$ in position $\SBC + m$.
Then, $i$ gets \ScF[\NUMCAND]{\SBC+m} points from $j$.
Because $j$ ranks all of \SB ahead of $i$, she gives at least
$\sum_{k=0}^{\SBC-1} \ScF[\NUMCAND]{k+m}$ points in total to \SB.
Hence, the difference in the number of points that $j$ gives to
an average candidate in \SB and the number of votes that $j$ gives to
$i$ is at least
\begin{align*}
\Big( \frac{1}{\SBC} \cdot \sum_{k=0}^{\SBC-1} & \ScF[\NUMCAND]{k+m} \Big)
- \ScF[\NUMCAND]{\SBC+m}
\; = \;
\frac{1}{\SBC} \cdot 
\sum_{k=0}^{\SBC-1} \left(\ScF[\NUMCAND]{k+m} - \ScF[\NUMCAND]{\SBC+m}\right).
\end{align*}

Because at least a $\SBM$ fraction of voters are in \SB,
the total advantage of an average candidate in \SB
resulting from voters in $B$ is at least
\[
\Delta_{\SB} \eqdef \frac{1}{\SBC} \cdot \SBM \cdot
\sum_{k=0}^{\SBC-1} \left(\ScF[\NUMCAND]{k+m} - \ScF[\NUMCAND]{\SBC+m}\right).
\]
\end{itemize}

We show that $\Delta_{\SB} > \Delta_i$, using condition
\eqref{eqn:shifted-limit-condition}.
Because \SCF[\NUMCAND] is monotone non-increasing,
and because $\SBC \geq Z$,
we get that
 \[
 \sum_{k=0}^{\SBC-1} \left(\ScF[\NUMCAND]{k+m} - \ScF[\NUMCAND]{\SBC+m}\right)
 \; \geq \; \sum_{k=0}^{Z-1} \left(\ScF[\NUMCAND]{k+m} - \ScF[\NUMCAND]{Z+m}\right)
 \; \stackrel{\eqref{eqn:shifted-limit-condition}}{>} \;
 \frac{2(1-z)}{z} \cdot      
 \sum_{k=\NUMCAND - Z}^{\NUMCAND-1} \left(1-\ScF[\NUMCAND]{k}\right).
 \]

Because $z > \half$ and \SCF[\NUMCAND] is monotone,
we get that
$\sum_{k=\NUMCAND - Z}^{\NUMCAND-1} (1-\ScF[\NUMCAND]{k})
\geq \half \sum_{k=\NUMCAND - \SBC}^{\NUMCAND-1} (1-\ScF[\NUMCAND]{k})$.
Hence,
\begin{align*}
\Delta_{\SB} & \; > \;
     \frac{1}{\SBC} \cdot \SBM \cdot \frac{1-z}{z} \cdot
        \sum_{k=\NUMCAND - \SBC}^{\NUMCAND-1} (1-\ScF[\NUMCAND]{k})
  \; \stackrel{\SBM \geq z}{\geq} \;
  \frac{1}{\SBC} \cdot (1-\SBM) \cdot
  \sum_{k=\NUMCAND - \SBC}^{\NUMCAND-1} (1-\ScF[\NUMCAND]{k})
   \; = \; \Delta_i.\QED
\end{align*}
\end{emptyproof}

We now wrap up the sufficiency portion of the proof of
Theorem~\ref{thm:main-generalized-borda}.
We distinguish two cases, based on the number of candidates \NUMCAND.
If $\NUMCAND < \NUMCANDH$,
then Lemma~\ref{lem:linear-distortion} implies an upper bound of
$\NUMCAND+1 \leq \NUMCANDH \leq \max(\NUMCAND_0, \frac{1}{1-z}) = O(1)$
on the expected distortion. 
Now assume that $\NUMCAND \geq \NUMCANDH$.
Recall that $\tr \eqdef \inf \Set{r}{\VProb{\Ball{\OPT}{r}} \geq y}$.
By Lemmas~\ref{lem:event-low-probability} and \ref{lem:event-implies-winner},
for any $r \geq \tr$, the probability that the
election's winner is outside \Ball{\OPT}{3r} is at most
$\frac{\e}{1-z} \cdot \Vol{r}$. The rest of the proof is similar to that of
Lemma~\ref{lem:linear-distortion}.
We again use that 
\[ \Expect[\CS]{\Cost{\Winner{\CS}}} \leq
  \Cost{\OPT} + \int_0^{\infty} \Prob[\CS]{\Dist{\OPT}{\Winner{\CS}} \geq r} \dd r, \]
and bound
\begin{align*}
  \int_0^{\infty} \Prob[\CS]{\Dist{\OPT}{\Winner{\CS}} \geq r} \dd r
& \; = \; \int_0^{3\tr} \Prob[\CS]{\Dist{\OPT}{\Winner{\CS}} \geq r} \dd r
  + \int_{3\tr}^{\infty} \Prob[\CS]{\Dist{\OPT}{\Winner{\CS}} \geq r} \dd r
\\ & \; \stackrel{\text{Lemmas~\ref{lem:event-low-probability}, \ref{lem:event-implies-winner}}}{\leq} \; 
       \int_0^{3\tr} 1 \; \dd r
     + \int_{3\tr}^{\infty} \frac{\e}{1-z} \cdot \Vol{r} \dd r
\\ & \; \leq \;
     3\tr + \frac{\e}{1-z} \cdot \int_0^{\infty} \Vol{r} \dd r
\\ & \; \stackrel{\eqref{eqn:expect-rewrite}}{=} \;
       3\tr + \frac{\e}{1-z} \cdot \Cost{\OPT}.
 \end{align*}

To upper-bound \tr, recall that at least a $1-\SBM$ fraction
of voters are outside of \Ball{\OPT}{\tr} or on the boundary.
Therefore, by Inequality~\eqref{eqn:opt-lower-bound},
$\Cost{\OPT} \geq \tr \cdot (1-\SBM)$.
Substituting this bound,
the expected cost of the winning candidate is at most
\[\left(1 + \frac{3}{1-\SBM} + \frac{\e}{1-z}\right) \cdot \Cost{\OPT}
= \left(1 + \frac{4}{1-\SBM}\right) \cdot \Cost{\OPT} = O(\Cost{\OPT}),\]
as $y$ depends only on the voting system \VS,
but not on the metric space or the number of candidates.
This completes the proof of sufficiency.

\subsubsection{Proof of Lemma~\ref{lem:shifted-comparison}}
\label{sec:shifted-comparison}

\begin{extraproof}{Lemma~\ref{lem:shifted-comparison}}
Because condition~\eqref{eqn:limit-condition} holds for all $y' > y$,
we may assume that $y \geq \half$.
Define $z_0 \eqdef \frac{5}{6} + \frac{y}{6}$,
and consider any $z \geq z_0$.
Fix $\NUMCAND \geq \NUMCAND_0$, and write
$Y := \Ceiling{y (\NUMCAND-1)}$ and $Z := \Ceiling{z (\NUMCAND-1)}$.
Let $m \leq (1-z) (\NUMCAND-1)$ be arbitrary.
We define
\begin{align*}
S_1 &\eqdef \sum_{k=0}^{\NUMCAND-1}
       \left( 1 - \ScF[\NUMCAND]{k} \right), \\
S_2 &\eqdef \sum_{k=0}^{m-1}
       \left( \ScF[\NUMCAND]{k} - \ScF[\NUMCAND]{Y} \right), \\
S_3 &\eqdef \sum_{k=m}^{Y-1}
       \left( \ScF[\NUMCAND]{k} - \ScF[\NUMCAND]{Y} \right).
\end{align*}
By monotonicity of \SCF[\NUMCAND],
\[
\sum_{k=0}^{Z-1} \left( \ScF[\NUMCAND]{m+k} - \ScF[\NUMCAND]{m+Z} \right)
  \; \geq \; S_3;
\]
furthermore,
$\sum_{k=\NUMCAND - Z}^{\NUMCAND-1} \left( 1 - \ScF[\NUMCAND]{k} \right)
\leq S_1$.
Therefore, it suffices to show that
$S_1 \leq \frac{z}{2 (1-z)} S_3$.
By condition~\eqref{eqn:limit-condition} and monotonicity of \SCF[\NUMCAND],
and because $y \geq \half$,
\begin{align*}
S_2 + S_3
& \; = \; \sum_{k=0}^{Y-1} \left( \ScF[\NUMCAND]{k} - \ScF[\NUMCAND]{Y} \right)
  \; \geq \; \frac{1-y}{y} \sum_{k=\NUMCAND-Y}^{\NUMCAND-1} \left( 1 - \ScF[\NUMCAND]{k} \right)
\; \geq \; \frac{1-y}{2y} S_1.
\end{align*}

To upper-bound $S_1$ in terms of $S_3$,
we show that the contribution of $S_2$ to the preceding sum is small,
and upper-bound $S_2$ in terms of $S_1 + S_3$.
Because $S_2 \leq (1-z) (\NUMCAND-1) \cdot \left( 1 - \ScF[\NUMCAND]{Y} \right)$,
using the monotonicity of \SCF[\NUMCAND], we can write
\begin{align*}
S_1 + S_3
& = \sum_{k=0}^{Y-1}
       \left( 1 - \ScF[\NUMCAND]{k} \right) + S_3
  + \sum_{k=Y}^{\NUMCAND-1} \left( 1 - \ScF[\NUMCAND]{k} \right)
\\ & \geq \sum_{k=m}^{Y-1}
       \left( 1 - \ScF[\NUMCAND]{Y} \right)
  + \sum_{k=Y}^{\NUMCAND-1} \left( 1 - \ScF[\NUMCAND]{Y} \right)
\\ & = (\NUMCAND-m) \cdot \left( 1 - \ScF[\NUMCAND]{Y} \right)
\\ & \geq z \cdot (\NUMCAND-1) \cdot \left( 1 - \ScF[\NUMCAND]{Y} \right)
\\ & \geq \frac{z}{1-z} \cdot S_2.
\end{align*}

Combining the preceding inequalities, we now obtain that
\[
\frac{1-y}{2y} \cdot S_1
\; \leq \; \frac{1-z}{z} (S_1 + S_3) + S_3
\; = \; \frac{1}{z} \cdot S_3 + \frac{1-z}{z} \cdot S_1.
\]
Solving for $S_1$, and using that the definition of $z_0$ ensures
$1-z \leq \frac{1-y}{6}$,
we now bound
\begin{align*}
S_1
& \; \leq \; \frac{2y}{z(1-y)-2y(1-z)} \cdot S_3
\; \leq \; \frac{2y}{4(1-z)} \cdot S_3
\; \leq \; \frac{z}{2(1-z)} \cdot S_3,
\end{align*}
completing the proof.
\end{extraproof}


\subsection{Necessity} \label{sec:necessity}

Next, we prove that the condition in
Theorem~\ref{thm:main-generalized-borda}
is also necessary for constant distortion.
We assume that the condition~\eqref{eqn:limit-condition}
does not hold, i.e., for every $y \in (0,1)$ and \NUMCANDZ,
there exists an $\NUMCAND \geq \NUMCANDZ$ such that
\begin{align}
& y \cdot \sum_{k=0}^{\Ceiling{y \cdot (\NUMCAND-1)}-1} \left( \ScF[\NUMCAND]{k} - \ScF[\NUMCAND]{\Ceiling{y \cdot (\NUMCAND-1)}} \right)
\; \leq \; (1-y) \cdot \sum_{k=\NUMCAND - \Ceiling{y \cdot (\NUMCAND-1)}}^{\NUMCAND-1} \left( 1 - \ScF[\NUMCAND]{k} \right).
\label{eqn:limit-violation}
\end{align}
We will show that the distortion of \VS is not bounded by any constant.

%
%

The high-level idea of the construction is as follows:
we define two tightly knit clusters \OC and \FC that are far away from each other.
\OC contains a large \OCM fraction of the population,
and thus should in an optimal solution be the one that the winner is
chosen from.
We will ensure that with probability at least
\half, the winner instead comes from \FC.
Because \FC is far from \OC, most of the population then is far from
the chosen candidate, giving much worse cost than optimal.

The metrics underlying \OC and \FC are as follows:
\FC will essentially provide an ``ordering,'' meaning that whichever
set of candidates is drawn from \FC, all voters in \FC
(and essentially all in \OC) agree on their ordering of the candidates.
This will ensure that one candidate from \FC will get a sufficiently
large fraction of first-place votes,
and will be ranked highly enough by voters from \OC, too.
\OC will be based on a large number \OCS of discrete locations \VP.
Their pairwise distances are chosen i.i.d.: as a result,
the rankings of voters are uniformly random, and there is no consensus
among voters in \OC on which of their candidates they prefer.
Because the vote is thus split, the best candidate from \FC will win
instead.

The following parameters (whose values are chosen with foresight) will
be used to define the metric space.
\begin{itemize}
\item Let $\APX > 1$ be any constant;
we will construct a metric space and number of candidates for
which the distortion is at least \APX.
\item Let $\FCM \in (0, \half)$ solve the quadratic
equation $\frac{2\FCM+1}{3\FCM} \cdot (1-\FCM) = 2 \APX - 1$.
A solution exists because at $\FCM = \half$,
the left-hand side is $\frac{2}{3} < 2\APX-1$;
it goes to infinity as $\FCM \to 0$,
while the right-hand side is a positive constant.
\FCM is the fraction of voters in the small cluster \FC.
\item Let $\OCM = 1-\FCM$ denote the fraction of voters in the large
  cluster \OC.
\item Let $\SCALE = \frac{1+\FCM}{\FCM}$ be the distance between
the clusters \FC and \OC.
(Each cluster will have diameter at most $2$.)
\item Let $\OCMU \geq \half + \frac{\OCM}{2} > \OCM$
satisfy $4\OCMU \cdot (1-\OCMU) < \OCM \cdot (1-\OCM)$;
such an \OCMU exists because the left-hand side goes to 0 as
$\OCMU \to 1$.
$\OCMU < 1$ is a high-probability upper bound on the fraction of
candidates that will be drawn from \OC.
\item Let $\NUMCANDZ = \frac{4}{\FCM^2} > 16$;
this is a lower bound on the number of candidates that ensures that
the actual fraction of candidates drawn from \OC is at most \OCMU with
sufficiently high probability.
\item Let $\NUMCAND \geq \NUMCANDZ$ be the \NUMCAND whose existence
  is guaranteed by the assumption \eqref{eqn:limit-violation}
  (for $y = \OCMU$ and \NUMCANDZ).
\item Let $\OCS = \NUMCAND^3$; this is the number of discrete
locations \VP we construct within the larger cluster \OC.
\end{itemize}

We now formally define the metric space consisting of two clusters:

\begin{definition} \label{def:two-cluster-metric}
The metric space consists of two clusters \OC and \FC.
\OC has \OCS discrete locations,
and \VPROB has a point mass of $\frac{\OCM}{\OCS}$ on each such location.
The total probability mass on \FC is
$\VProb{\FC} = 1-\OCM$,
distributed uniformly over the interval $[1,2]$.
Locations in \FC are identified by $\FCP \in [1,2]$.
The distances are defined as follows:

\begin{enumerate}
\item For each distinct pair $\VP, \VPP \in \OC$,
  the distance \Dist{\VP}{\VPP} is drawn independently uniformly at
  random from $[1,2]$.
\item For each distinct pair $\FCP, \FCPP \in \FC$ of locations,
  the distance is defined to be $\Dist{\FCP}{\FCPP} := \min(\FCP,\FCPP)$.
\item Partition $\FC = [1,2]$ into $\OCS!$ disjoint intervals
  $I_{\pi}$ of length $1/\OCS!$ each,
  one for each permutation of the \OCS locations in \OC.
  For $\VP \in \OC$ and $\FCP \in I_{\pi}$,
  let $\pi^{-1}(\VP)$ be the position of \VP in $\pi$,
  and define the distance between \VP and \FCP to be
  $\Dist{\VP}{\FCP} = \SCALE + \frac{\FCP}{4} + \frac{\pi^{-1}(\VP)}{\OCS!}$.
\end{enumerate}
\end{definition}

\begin{proposition} \label{prop:metric}
Definition~\ref{def:two-cluster-metric} defines a metric.
\end{proposition}

\begin{proof}
Non-negativity, symmetry, and indiscernibles hold by definition.
Because all distances within clusters are in $[1,2]$, and distances
across clusters are more than 2, the triangle inequality holds for
all pairs $\VP, \VPP \in \OC$ and all pairs $\FCP, \FCPP \in \FC$.

Because $\Dist{\VP}{\FCP} \in [\SCALE, \SCALE+1]$ for all $\VP \in \OC$
and $\FCP \in \FC$,
and distances within \OC or \FC are at least 1,
there can be no shorter path than the direct one
between any $\VP \in \OC$ and $\FCP \in \FC$.
Therefore, the triangle inequality is satisfied.
\end{proof}

Now consider a (random) set \CS of 
\NUMCAND candidates, drawn i.i.d.~from \VPROB.
We are interested in the event that the resulting slate of candidates
is highly representative of the voters, in the following sense.

\begin{definition} \label{def:representative-candidates}
Let \CS be the (random) set of \NUMCAND candidates drawn from \VPROB.
Let \Event be defined as the conjunction of the following:
\begin{enumerate}
\item For each location $\VP \in \OC$,
  the set \CS contains at most one candidate from \VP.
\item At least a $\frac{\FCM}{2}$ fraction of candidates in \CS
  is drawn from \FC (and thus at most an \OCMU fraction of
    candidates are from \OC).
\item At least an $\frac{\OCM}{2}$ fraction of candidates in \CS
  is drawn from \OC.
\item No pair $\FCP, \FCPP \in \FC \cap \CS$ has $\Abs{\FCP - \FCPP} < \frac{1}{(\OCS-1)!}$.
\end{enumerate}
\end{definition}

Lemma~\ref{lem:event-high-probability} uses standard tail bounds to
show that \Event happens with probability at least \half;
then, Lemma~\ref{lem:event-wrong-winner} shows that whenever \Event happens,
the winner is from \FC. 

\begin{lemma} \label{lem:event-high-probability}
\Event happens with probability at least \half.
\end{lemma}
\begin{proof}
We upper-bound the probability of the complement of each of the four
constituent sub-events.
\begin{enumerate}
\item For each of the at most $\NUMCAND^2$ pairs of candidates,
the probability that they are both drawn from the same location is
at most $\OCM/\OCS \leq 1/\NUMCAND^3$.
By a union bound over all pairs, the probability that any location has
at least two pairs is at most $1/\NUMCAND$.

\item Let the random variable $X$ be the number of candidates drawn
from \FC.
Then, $\Expect{X} = \FCM \cdot \NUMCAND$,
and $X$ is a sum of i.i.d.~Bernoulli random variables.
By the Hoeffding bound 
$\Prob{X < (\FCM - \epsilon) \NUMCAND} \leq \exp(-2 \epsilon^2 \NUMCAND)$,
with $\epsilon = \FCM/2$, we obtain that the fraction of
candidates from \FC is too small with probability at most
$\exp(\frac{-\FCM^2}{2} \cdot \NUMCAND)
\leq \exp(\frac{-\FCM^2}{2} \cdot \NUMCANDZ) = \frac{1}{\e^2}$.

\item The proof is essentially identical to the previous case (except
because $\OCM \geq \FCM$, the bounds are even stronger),
so this event happens with probability at least $\frac{1}{\e^2}$ as
well.

\item
Consider all intervals of $[1,2]$ of length $\frac{2}{(\OCS-1)!}$,
starting at $1+\frac{k}{(\OCS-1)!}$ for some 
$k = 0, 1, \ldots, (\OCS-1)!-2$.
If $\FCP,\FCPP$ with $\Abs{\FCP-\FCPP} \leq \frac{1}{(\OCS-1)!}$ existed,
they would both be contained in at least one such interval
(because the interval length is twice as long as the distance).

For any of the $(\OCS-1)!-1$ intervals $I$, the probability that a specific
pair of candidates is drawn from $I$ is at most
$\frac{4}{((\OCS-1)!)^2}$.
By a union bound over all (at most $\NUMCAND^2$) pairs of
candidates and all intervals,
the probability that any pair is drawn from  any interval $I$ is at
most $\frac{4\NUMCAND^2}{(\OCS-1)!} \leq \frac{1}{\NUMCAND}$.
\end{enumerate}
Because $\NUMCAND \geq 9$, a union bound shows that \Event happens
with probability at least $\half$.
\end{proof}

\begin{lemma} \label{lem:event-wrong-winner}
Whenever \Event happens, the winning candidate is from \FC.
\end{lemma}

\begin{proof}
Let \FCMC be the actual number of candidates drawn from \FC,
and $\OCMC = \NUMCAND - \FCMC$ the number of candidates drawn from \OC.
Because we assumed that \Event happened,
$\FCMC \geq \frac{\FCM}{2} \cdot \NUMCAND$
and $\OCMC \leq \OCMU \cdot \NUMCAND$.
Let \OCC be the set of candidates drawn from \OC.
Under \Event, \OCC contains at most one candidate from each location
$\VP \in \OC$.
As a result, because the random distances within \OC are distinct with
probability $1$, 
there will be no ties in the rankings of any voters.

Let $\hat{\imath}$ be the candidate from \FC with smallest value \FCPH.
With probability $1$, the \FCP value of $\hat{\imath}$ is unique.
Consider some arbitrary candidate $i \in \OCC$ from location \VPP.
We calculate the contributions to $\hat{\imath}$ and $i$
from voters in \FC and in \OC separately,
and show that $\hat{\imath}$ beats $i$.
Because this holds for arbitrary $i$,
the candidate $\hat{\imath}$ or another candidate from $\FC$ wins.

\begin{enumerate}
\item We begin with points given out by voters in \FC.
By definition of the distances within \FC,
$\hat{\imath}$ is ranked first by all voters in \FC.

Voters in $I_{\pi}$ rank the candidates from \OC according to their
order in $\pi$.
For each ordering of \OCC, exactly a $\frac{1}{\OCMC!}$ fraction of
permutations induces that ordering.
In particular, for each $k \in {1, \ldots, \OCMC}$,
exactly a $1/\OCMC$ fraction of voters
places $i$ in position $k + \FCMC$.
Thus, $i$ obtains a total of
$(1-\OCM) \cdot \sum_{k=\NUMCAND-\OCMC}^{\NUMCAND-1} \frac{1}{\OCMC} \cdot \ScF[\NUMCAND]{k}$
points from voters in \FC.
Overall, $\hat{\imath}$ obtains an advantage of at least
\begin{align*} 
\Delta_{\FC}
& \; = \; (1-\OCM) \cdot \left(\ScF[\NUMCAND]{0} -
   \frac{1}{\OCMC} \cdot \sum_{k=\NUMCAND-\OCMC}^{\NUMCAND-1} \ScF[\NUMCAND]{k}
   \right)
 \; = \; (1-\OCM) \cdot \frac{1}{\OCMC} \cdot 
     \sum_{k=\NUMCAND-\OCMC}^{\NUMCAND-1} (1-\ScF[\NUMCAND]{k}).
\end{align*}
   
\item Next, we analyze the number of points given out by voters in \OC.
The distance from any voter location $\VP \in \OC$ to $\hat{\imath}$ is
at most $\SCALE + \frac{\FCPH}{4} + \frac{\NUMCAND}{\OCS!}$.
Under \Event, no other candidate from \FC can be at a location
$\FCP \leq \FCPH + \frac{1}{(\OCS-1)!}$;
therefore, the distance from any voter location $\VP \in \OC$
to any other candidate $\FCP \in \FC$ is at least
\[
  \Dist{\VP}{\FCP}
  \; \geq \; \SCALE + \frac{\FCP}{4}
  \; \geq \; \SCALE + \frac{\FCPH}{4} + \frac{1}{4(\OCS-1)!}
  \; > \; \SCALE + \frac{\FCPH}{4} + \frac{\NUMCAND}{\OCS!}
  \; \geq \; \Dist{\VP}{\hat{\imath}},
\]

so all voters in \OC prefer $\hat{\imath}$ over any other candidate from \FC.
Hence, $\hat{\imath}$ obtains at least
$\OCM \cdot \ScF[\NUMCAND]{\OCMC}$
points combined from voters in \OC.

To analyze the votes from voters in \OC for candidates from \OC,
we first notice that \Event and the draw of candidates
are independent of the distances within \OC.
Hence, even conditioned on \Event,
the distances \Dist{\VP}{\VPP} between locations in
\OC are i.i.d.~uniform from $[1,2]$.
In particular, each location $\VP \in \OC$ ranks the candidates in
\OCC in uniformly random order.
Furthermore, for two locations $\VP \neq \VPP$,
the rankings of \OCC are independent;
the reason is that they are based on disjoint vectors of distances
$(\Dist{\VP}{i})_{i \in \OCC}, (\Dist{\VPP}{i})_{i \in \OCC}$.
We use this independence to apply tail bounds.
Let \VPP be the location of $i$.
Voters rank $i$ as follows:

\begin{itemize}
\item Among locations \VP without a candidate of their own,
  in expectation,
  a $1/\OCMC$ fraction of voters will rank $i$ in position $k$,
  for each $k = 0, \ldots, \OCMC-1$.
\item Among the $\OCMC-1$ locations $\VP \neq \VPP$
  with a candidate of their own, in expectation,
  a $1/(\OCMC-1)$ fraction of voters will rank $i$ in position $k$,
  for each $k = 1, \ldots, \OCMC -1$.
\item Voters at \VPP will rank $i$ in position 0.
\end{itemize}

For each $k$, let the random variable $X_k$ be the number of locations
that rank $i$ in position $k$.
By the preceding arguments, $\Expect{X_k} = \frac{\OCS}{\OCMC}$,
and $X_k$ is a sum of \OCS independent (not i.i.d.)
Bernoulli random variables.
Hence, by the Hoeffding bound,
the probability that more than a $\frac{2}{\OCMC}$
fraction of voters rank $i$ in position $k$
is at most $2\exp(-2 \cdot \frac{1}{\OCMC^2} \cdot \OCS) \leq 2\exp(-\NUMCAND)$.
By a union bound over all candidates $i \in \OCC$ and all values
$k = 0, \ldots, \OCMC - 1$,
with high probability, for all $i$ and $k$,
the fraction of voters (in \OC) ranking $i$ in position $k$ is at most
$\OCM \cdot \frac{2}{\OCMC}$.
Because the total fraction of voters in \OC is \OCM,
any excess votes for some (early) positions $k$ must be compensated
by fewer votes for other (late) positions $k'$.
Relaxing the constraint that the number of votes for each position $k$
must be non-negative, we can upper-bound the total points for $i$
by assuming that each of the positions $k=0, \ldots, \OCMC-2$ receives
twice the expected number of votes,
while position $k=\OCMC-1$ receives a negative number of votes
that compensates for the excess votes.
Then, the advantage for $i$ over $\hat{\imath}$ from votes from \OC
is at most
\begin{align*}
\Delta_{\OC}
& \; \eqdef \;
\OCM \cdot \left(
    \sum_{k=0}^{\OCMC-2} \frac{2}{\OCMC} \cdot \ScF[\NUMCAND]{k}
  + \frac{2-\OCMC}{\OCMC} \cdot \ScF[\NUMCAND]{\OCMC-1}
- \ScF[\NUMCAND]{\OCMC} \right)
\\ & \; = \;
\frac{\OCM}{\OCMC} \cdot \bigg( \sum_{k=0}^{\OCMC-2} 
\big( 2 \ScF[\NUMCAND]{k}
     - \ScF[\NUMCAND]{\OCMC-1}
     - \ScF[\NUMCAND]{\OCMC} \big) + \ScF[\NUMCAND]{\OCMC-1} - \ScF[\NUMCAND]{\OCMC} \bigg)
\\ & \; \stackrel{\SCF[\NUMCAND] \text{ monotone}}{\leq} \;
 \frac{2\OCM}{\OCMC} \cdot
 \sum_{k=0}^{\OCMC-1} (\ScF[\NUMCAND]{k} - \ScF[\NUMCAND]{\OCMC}).
\end{align*}
\end{enumerate}

Finally, we can bound

\begin{align*}
\Delta_{\OC} \cdot \OCMC
& \leq
2 \OCM \cdot \sum_{k=0}^{\OCMC-1} 
  (\ScF[\NUMCAND]{k} -  \ScF[\NUMCAND]{\OCMC})
\\ & \stackrel{\OCM \leq \OCMU, \, \SCF[\NUMCAND] \text{ mon.}}{\leq}
2 \OCMU \cdot \sum_{k=0}^{\Ceiling{\OCMU (\NUMCAND-1)}-1} 
  (\ScF[\NUMCAND]{k} -  \ScF[\NUMCAND]{\Ceiling{\OCMU (\NUMCAND-1)}})
\\ & \stackrel{\eqref{eqn:limit-violation}, \text{ Def.~of } \NUMCAND}{\leq}
2 (1-\OCMU) \cdot 
  \sum_{k=\NUMCAND-\Ceiling{\OCMU (\NUMCAND-1)}}^{\NUMCAND-1} (1-\ScF[\NUMCAND]{k})
\\ & \stackrel{\SCF[\NUMCAND] \text{ mon.}}{\leq}
2 (1-\OCMU) \cdot \frac{\OCMU \cdot (\NUMCAND-1)}{\OCMC}
  \sum_{k=\NUMCAND-\OCMC}^{\NUMCAND-1} (1-\ScF[\NUMCAND]{k})
\\ & \stackrel{\OCMC \geq \OCM \cdot \NUMCAND/2}{\leq}
2 (1-\OCMU) \cdot \frac{2\OCMU}{\OCM}
  \sum_{k=\NUMCAND-\OCMC}^{\NUMCAND-1} (1-\ScF[\NUMCAND]{k})
\\ & \stackrel{\text{Def.~of } \OCMU}{<}
(1-\OCM) \cdot \sum_{k=\NUMCAND-\OCMC}^{\NUMCAND-1} (1-\ScF[\NUMCAND]{k})
\\ & = \Delta_{\FC} \cdot \OCMC.
\end{align*}

Thus, $\hat{\imath}$ beats all candidates drawn from \OC,
and the winner will be from \FC.
\end{proof}

Using the preceding lemmas, the proof of
necessity is almost complete.
Consider the metric space with all the parameters as defined above.
By Lemmas~\ref{lem:event-high-probability} and \ref{lem:event-wrong-winner},
with probability at least \half, the winner is from \FC.
The social cost of any candidate from \FC is at least
$\FCM \cdot 0 + (1-\FCM) \cdot (\SCALE+1)$.
On the other hand, the social cost of any candidate from \OC is at
most $(1-\FCM) \cdot 2 + \FCM \cdot (\SCALE+1) = 3$.
The distortion in this case is thus at least
\begin{align*}
\frac{(1-\FCM) \cdot (\SCALE+1)}{3}
& \; = \;
\frac{(2\FCM+1) \cdot (1-\FCM)}{3\FCM}
  \; = \; 2\APX-1.
\end{align*}
In the other case (when \Event does not occur --- this happens with
probability at most \half),
the distortion is at least $1$, so that the expected distortion is at
least $\half (2\APX-1) + \half \cdot 1 = \APX$.

\section{Proof of Corollary~\ref{cor:generalized-borda-constant}}
\label{sec:generalized-borda-constant-proof}

\begin{extraproof}{Corollary~\ref{cor:generalized-borda-constant}}
For the first part of the corollary,
assume that \SCF is not constant on $(0,1)$.
The intuition is that in that case,
the sum on the left-hand side of \eqref{eqn:limit-condition}
(for sufficiently large $y$) will be $\Omega(\NUMCAND)$,
while the sum on the right-hand side is obviously at most \NUMCAND. 
By making $y$ a constant close enough to $1$,
we can dominate the constant from $\Omega$,
and thus ensure that the inequality \eqref{eqn:limit-condition} holds.
Then, the constant distortion follows from
Theorem~\ref{thm:main-generalized-borda}.

More precisely, let $0 < \ScFLB < \ScFUB < 1$ be such that $\ScF{\ScFLB} > \ScF{\ScFUB}$.
Let $\delta \eqdef \ScF{\ScFLB} - \ScF{\ScFUB}$
and $y \eqdef \max(\ScFUB, 1-\frac{\delta \ScFLB}{8}) \in (0,1)$.
Let \NUMCANDZ be such that for all $n \ge n_0$, we have
\begin{align*}
\ScF[\NUMCAND]{\Floor{\ScFLB \cdot (\NUMCAND-1)}} & \geq \ScF{\ScFLB} - \delta/4, &
 \ScF[\NUMCAND]{\Ceiling{\ScFUB \cdot (\NUMCAND-1)}} & \leq \ScF{\ScFUB} + \delta/4,\\
\Floor{\ScFLB \cdot (\NUMCAND-1)} & \geq  \frac{\ScFLB \NUMCAND}{2}, & 
\Ceiling{y \cdot (\NUMCAND -1)} & \leq 2 y \NUMCAND.
\end{align*}
Such an \NUMCANDZ exists by the consistency of \VS and basic integer arithmetic.
Then, for all $\NUMCAND \geq \NUMCANDZ$,
\begin{align*}
& y \cdot \sum_{k=0}^{\Ceiling{y \cdot (\NUMCAND-1)}-1}
  \left(  \ScF[\NUMCAND]{k}
        - \ScF[\NUMCAND]{\Ceiling{y \cdot (\NUMCAND-1)}} \right)
\\ & \geq y \cdot \sum_{k=0}^{\Floor{\ScFLB \cdot (\NUMCAND-1)}-1}
  \left(  \ScF[\NUMCAND]{\Floor{\ScFLB \cdot (\NUMCAND-1)}}
        - \ScF[\NUMCAND]{\Ceiling{\ScFUB \cdot (\NUMCAND-1)}} \right)
\\ & \geq y \cdot \sum_{k=0}^{\Floor{\ScFLB \cdot (\NUMCAND-1)}-1} (\delta/2)
\\ & \geq \frac{1}{4} \cdot y \cdot \ScFLB \cdot \NUMCAND \cdot \delta
\\ & \geq 2 y \cdot (1-y) \cdot \NUMCAND
\\ & > (1-y) \cdot \sum_{k=\NUMCAND-\Ceiling{y \cdot (\NUMCAND-1)}}^{\NUMCAND-1} \left( 1 - \ScF[\NUMCAND]{k} \right).
\end{align*}

Because the condition~\eqref{eqn:limit-condition} is satisfied,
Theorem~\ref{thm:main-generalized-borda} implies constant distortion.

\medskip

For the second part of the corollary,
assume that $\ScF{x} = \SCC < 1$ for all $x \in (0,1)$.
Let $y \in (0,1)$ be arbitrary.
We will show that for sufficiently large \NUMCAND,
the condition~\eqref{eqn:limit-condition} is violated.

The intuition is that the sum on the right-hand side of
\eqref{eqn:limit-condition} consists of terms that will in the
limit be $1-\SCC > 0$,
while the left-hand side is a sum in which each term converges to 0.
Thus, never mind how large the constant $y < 1$ is,
the factors of $y$ and $1-y$ will eventually not be enough to make the
left-hand side larger than the right-hand side.
Making this intuition precise requires some care:
while the functions \SCF[\NUMCAND] converge to \SCF,
we did not assume that they do so \emph{uniformly}. 
To deal with this issue, we will consider consistency with \SCF at two points
$\ScFSm$ and $1-\ScFSm$ only (with $\ScFSm$ being a very small constant),
and use monotonicity of each \SCF[\NUMCAND] to bound the remaining
terms.
The terms of the sum corresponding to points to the left of $\ScFSm$ and to
the right of $1-\ScFSm$ can then not be bounded, but there are few enough
of them that we still obtain the desired inequality.
More specifically, let $\ScFSm \in (0,1)$ be a sufficiently small constant such that $\ScFSm < \min(y,1-y)$ and 
\[
  \delta \; := \;
  \frac{(1-y) \cdot (1-\SCC) - \ScFSm}{1+ 3 y - 4 \ScFSm}
  \; > \; 0.
\]

Such a $\ScFSm$ exists, since both the numerator and denominator tend to strictly positive numbers as $\ScFSm \to 0$.
Recall that $\ScF{x} = \SCC$ for all $x \in (0,1)$.
Let \NUMCANDZ be such that for all $\NUMCAND \geq \NUMCANDZ$,
\begin{align*}
\ScF[\NUMCAND]{\Floor{\ScFSm \cdot (\NUMCAND-1)}} & \leq \ScF[\NUMCAND]{\Ceiling{\frac{\ScFSm}{2} \cdot (\NUMCAND-1)}} \leq \SCC + \delta, \\
  \ScF[\NUMCAND]{\Ceiling{(1-\ScFSm) (\NUMCAND-1)}} & \geq \ScF[\NUMCAND]{\Floor{(1-\frac{\ScFSm}{2}) (\NUMCAND-1)}} \geq \SCC - \delta,\\
\Ceiling{y \cdot (\NUMCAND -1)} - \Floor{\ScFSm \cdot (\NUMCAND -1)} & \leq 2 (y-\ScFSm) (\NUMCAND -1).
\end{align*}
Such an \NUMCANDZ exists by basic integer arithmetic and the consistency of \VS applied at $x = \ScFSm/2$ and $x = 1-\ScFSm/2$.

Writing $\ScFSmC = \Floor{\gamma (\NUMCAND-1)}$ and
$\ScFSmCp = \Ceiling{(1-\gamma) (\NUMCAND-1)}$, we get
\begin{align*}
y \cdot \sum_{k=0}^{\Ceiling{y \cdot (\NUMCAND-1)}-1}
  \left( \ScF[\NUMCAND]{k} - \ScF[\NUMCAND]{\Ceiling{y \cdot (\NUMCAND-1)}} \right)
 & \; \leq \; 
  y \cdot \left(
  \sum_{k=0}^{\ScFSmC-1} 1 + 
  \sum_{k=\ScFSmC}^{\Ceiling{y \cdot (\NUMCAND-1)}-1}
  \left( \ScF[\NUMCAND]{\ScFSmC} - \ScF[\NUMCAND]{\ScFSmCp} \right)
  \right)
\\ & \; \leq \; 
  y \cdot \left(\ScFSm \cdot (\NUMCAND-1) + 2 \cdot (y-\ScFSm) \cdot (\NUMCAND-1) \cdot 2\delta \right)
\\ & \; = \; 
  y \cdot (\NUMCAND-1) \cdot \left(\ScFSm + (y-\ScFSm) \cdot 4\delta \right).
\end{align*}


The first inequality uses $y < 1-\ScFSm$ and the monotonicity of \SCF[\NUMCAND],
and the second inequality uses the bounds obtained from consistency of \SCF[\NUMCAND] with respect to \SCF.
To bound the right-hand side of \eqref{eqn:limit-condition},
\begin{align*}
(1-y) \cdot \sum_{k=\NUMCAND-\Ceiling{y \cdot (\NUMCAND-1)}}^{\NUMCAND-1}
  \left( 1 - \ScF[\NUMCAND]{k} \right)
& \; \geq \; (1-y) \cdot \sum_{k=\NUMCAND-\Ceiling{y \cdot (\NUMCAND-1)}}^{\NUMCAND-1}
  \left( 1 - \ScF[\NUMCAND]{\ScFSmC} \right)
\\ 
& \; \geq \; (1-y) \cdot y  \cdot (\NUMCAND-1) \cdot (1-\SCC-\delta).
\end{align*}

The first inequality again used monotonicity of \SCF[\NUMCAND],
and the second used the bounds obtained from the consistency of \SCF[\NUMCAND] with respect to \SCF.
Canceling the common term $y (\NUMCAND-1)$ between the left-hand side and right-hand side,
the right-hand side of \eqref{eqn:limit-condition} is at least as large as the left-hand side whenever
$(1-y) \cdot (1-\SCC-\delta) \geq  \ScFSm + (y-\ScFSm) \cdot 4\delta$.
Solving for $\delta$, this is equivalent to
\[ \delta \leq \frac{(1-y) \cdot (1-\SCC) - \ScFSm}{1+ 3 y - 4 \ScFSm}, \]
which is exactly ensured by our choice of $\ScFSm$ and $\delta$.
This completes the proof.
\end{extraproof}



\section{Conclusions} \label{sec:conclusions}

When candidates are drawn i.i.d.~from the voter distribution,
we showed that whether a positional voting system \VS has
expected constant distortion can be almost fully characterized by its
limiting behavior. In particular,
if the limiting scoring rule is not constant on $(0,1)$,
then \VS has constant expected distortion;
if the limiting scoring rule is a constant other than 1 on $(0,1)$,
then \VS has super-constant expected distortion.
A more subtle condition depending on the ``rate of convergence'' to
the limit rule completes the characterization.

Our Theorem~\ref{thm:main-generalized-borda} currently does not
characterize the order of growth of the distortion.
With some effort, the proof could likely be adapted to the
case where the $y$ in the theorem is a function $y(\NUMCAND)$,
which would allow us to characterize the rate at which the distortion
grows with \NUMCAND.

For specific voting systems, the proof of
Theorem~\ref{thm:main-generalized-borda} can often be adapted
to give tighter bounds.
For example, straightforward modifications of the proof can be used to
show that the distortion of $k$-approval or $k$-veto
(where each voter can veto $k$ candidates) for constant $k$ grow as
$\Omega(\NUMCAND)$.
This matches the $O(\NUMCAND)$ upper bound from
Lemma~\ref{lem:linear-distortion},
giving a tight analysis of the distortion of these voting systems.
Similarly, the sufficiency proof can be adapted
to show that the distortion of Borda Count is at most 16,
for all metric spaces and all \NUMCAND.
When the number of candidates grows large enough,
the expected distortion is in fact bounded by 10.


Our results indicate that if one is concerned about systematic, and
possibly adversarial, bias in which candidates run for office,
randomizing the slate of candidates may be part of a solution approach.
Such an approach can be considered as a step in the direction of 
\emph{lottocracy} and \emph{sortition}
\cite{dowlen:sortition,guerrero:lottocracy,landemore:lottocracy},
in which office holders are directly chosen at random from the
population.
Pure lottocracy does well in terms of \emph{representativeness} of
office holders,
but one of its main drawbacks is the potential lack of competency.
As a broader direction for future research, our work here suggests
devising models that capture the tension between these
two objectives, and would allow for the design of hybrid mechanisms
that navigate the tradeoff successfully.



\subsubsection*{Acknowledgments}
Part of this work was done while Yu Cheng was a student at the University of Southern California.
Yu Cheng was supported in part by Shang-Hua Teng's Simons Investigator Award.
Shaddin Dughmi was supported in part by NSF CAREER Award CCF-1350900 and NSF grant CCF-1423618.
David Kempe was supported in part by NSF grants CCF-1423618 and IIS-1619458.
We would like to thank anonymous reviewers for useful feedback.

\bibliographystyle{plain}
\bibliography{../bibliography/names,../bibliography/conferences,../bibliography/bibliography,../bibliography/voting}

%

\end{document}